\documentclass[11pt]{article}

\usepackage[noline,ruled,noend,linesnumbered]{algorithm2e}
\usepackage{amsfonts,amsmath,amssymb,amsthm}
\usepackage{fullpage}
\usepackage{xcolor}
\usepackage{cite}

\newtheorem{thm}{Theorem}[section]
\newtheorem{claim}[thm]{Claim}

\newtheorem{lemma}[thm]{Lemma}
\newtheorem{prop}[thm]{Proposition} 
\newtheorem{fact}[thm]{Fact}
\theoremstyle{definition}
\newtheorem{defn}[thm]{Definition}
\theoremstyle{plain}

\DeclareMathOperator*{\E}{E}
\DeclareMathOperator*{\Var}{Var}

\DeclareMathOperator*{\argmin}{argmin}
\newcommand{\bfx}{X}
\newcommand{\bfy}{Y}
\newcommand{\core}{\mathrm{core}}
\newcommand{\dist}{{\rm dist}}
\newcommand{\distH}{\dist_{\rm Ham}}
\newcommand{\distp}{\dist_{p}}

\newcommand{\disttwo}{\dist_{2}}
\newcommand{\dTV}{\mathrm{d}_{\mathrm{TV}}}

\newcommand{\Inf}{\mathrm{Inf}}
\newcommand{\poly}{{\rm poly}}
\newcommand{\R}{\mathbb{R}}
\newcommand{\supp}{\mathrm{supp}}

\newcommand{\EstInf}{\textsc{EstimateInf}\xspace}

\newcommand{\Vondrak}{Vondr{\'a}k\xspace}

\newcommand{\calF}{{\cal F}}

\newcommand{\calP}{{\cal P}}

\title{Testing submodularity and other properties of valuation functions}

\author{
	Eric Blais\\
	School of Computer Science\\
	University of Waterloo\\
	\texttt{eric.blais@uwaterloo.ca}
	\and
	Abhinav Bommireddi\\
	School of Computer Science\\
	University of Waterloo\\
	\texttt{vabommir@uwaterloo.ca}
}

\begin{document}

\maketitle

\begin{abstract}
We show that for any constant $\epsilon > 0$ and $p \ge 1$, it is possible to distinguish functions $f : \{0,1\}^n \to [0,1]$ that are submodular from those that are $\epsilon$-far from every submodular function in $\ell_p$ distance with a \emph{constant} number of queries. 

More generally, we extend the testing-by-implicit-learning framework of Diakonikolas et al.~(2007) to show that every property of real-valued functions that is well-approximated in $\ell_2$ distance by a class of $k$-juntas for some $k = O(1)$ can be tested in the $\ell_p$-testing model with a constant number of queries. This result, combined with a recent junta theorem of Feldman and \Vondrak (2016), yields the constant-query testability of submodularity. It also yields constant-query testing algorithms for a variety of other natural properties of valuation functions, including fractionally additive (XOS) functions, OXS functions, unit demand functions, coverage functions, and self-bounding functions.
\end{abstract}

\section{Introduction}

Property testing is concerned with approximate decision problems of the following form: given oracle access to some function $f : \mathcal{X} \to \mathcal{Y}$ and some fixed property $\mathcal{P}$ of such functions, how many oracle calls (or queries) to $f$ does a bounded-error randomized algorithm need to distinguish the cases where $f$ has the property $\mathcal{P}$ from the case where $f$ is $\epsilon$-far---under some appropriately defined metric---from having the same property? 
Remarkably, many natural properties of functions can be tested with a number of queries that is \emph{independent} of the size of the function's domain. For example, for any constant $\epsilon > 0$ and $t \ge 1$, a constant number of queries suffices to test whether a Boolean function $f : \{0,1\}^n \to \{0,1\}$ is linear~\cite{BLR93}; a polynomial of degree at most $t$~\cite{RS96}; a $t$-junta~\cite{FKRSS04,Blais09}; a monomial~\cite{PRS02}; computable by a Boolean circuit of size $t$~\cite{DLMORSW07}; or a linear threshold function~\cite{MORS10}.

In this work, we consider the problem of testing properties of bounded \emph{real-valued} functions over the Boolean hypercube. In particular, are there natural examples of such properties that are testable with a constant number of queries? This question is best considered in the \emph{$\ell_p$ testing} framework introduced by Berman, Raskhodnikova, and Yaroslavtsev~\cite{BRY14}. In this setting, the distance between a function $f : \{0,1\}^n \to [0,1]$ and some property $\mathcal{P}$ of these functions is $\dist_p(f,\mathcal{P}) = \inf_{g \in \mathcal{P}} \| f - g \|_p$.

\subsection{Testing properties of valuation functions}

Natural properties of bounded real-valued Boolean functions have been studied extensively in the context of valuation functions in algorithmic game theory. For a sequence of $n$ goods labeled with the indices $1,\ldots,n$, we can encode the value of each subset of these goods to some agent with a function $f : \{0,1\}^n \to [0,1]$ by setting $f(x)$ to be the (possibly normalized) value of the subset $\{ i \in [n] : x_i = 1\}$ to the agent. Such a valuation function $f$ is
\begin{description}
	\item[Additive] if there are weights $w_1,\ldots,w_n$ such that $f(x) = \sum_{i : x_i = 1} w_i$;

	\item[{\rm a} Coverage function] 
	if there exists a universe $U$, non-negative weights $\{w_u\}_{u \in U}$, and subsets $A_1,..., A_n \subseteq U$ such that $f(x) = \sum_{u \in \bigcup_{i : x_i = 1}A_i} w_u$. 

	\item[Unit demand] 
	if there are weights $w_1,\ldots,w_n$ such that $f(x) = \max\{ w_i : x_i = 1\}$;

	\item[OXS] 
	if there are $k \ge 1$ unit demand functions $g_1,\ldots,g_k$ such that $f(x) = \max\{ g_1(x^{(1)}),\ldots,g_k(x^{(k)})\}$ where the maximum is taken over all $x^{(1)},\ldots,x^{(k)}$ such that for every $i \in [n]$, 
	$x_i = \sum_{j=1}^k x^{(j)}_{\,i}$;
	
	\item[Gross Substitutes] if for any $p' \le p \in \R^n$ and any $x,x'$ that maximize $f(x) - \sum_{i : x_i = 1} p_i$ and $f(x') - \sum_{i : x'_i = 1} p'_i$, respectively, every $j \in [n]$ for which $x_j = 1$ and $p_j = p'_j$ also satisfies $x'_j = 1$;
	
	\item[Submodular] if $f(x) + f(y) \ge f(x \wedge y) + f(x \vee y)$ for every $x,y \in \{0,1\}^n$, where $\wedge$ and $\vee$ are the bitwise AND and OR operations;

	\item[Fractionally subadditive (XOS)] iff there are non-negative real valued weights $\{w_{ij}\}_{i,j \le n}$ such that $f(x) = \max_i \sum_{j} w_{ij}\cdot x_j$;
	
	\item[Self-bounding] if $f(x) \geq \sum_i(f(x) - \min_{x_i} f(x))$, where $ \min_{x_i} f(x) = \min\{f(x), f(x \oplus e_i)\}$ and $\oplus$ is the bitwise XOR operator; and

	\item[Subadditive] if $f(x \cup y) \le f(x) + f(y)$ for every $x,y \in \{0,1\}^n$.
\end{description}
Each of these properties enforces some structure on valuation functions, and much work has been devoted to better understanding these structures (and their algorithmic implications) by studying the properties through the lenses of learning theory~\cite{BH11,BCIW12,FK14}, optimization~\cite{Fei09,FMV11}, approximation~\cite{FV16,FV15}, and sketching~\cite{BDFKNR12}.
The problem of testing whether an unknown valuation function satisfies one of these properties offers another angle from which we can learn more about the structure imposed on the functions that satisfy these properties. 

Indeed, there has already been some recent developments on the study of testing these properties. Notably, Seshadhri and \Vondrak~\cite{SV11} initiated the study of testing submodularity for functions over the hypercube and showed that in the setting where we measure the distance to submodularity in terms of Hamming distance (rather than $\ell_p$ distance), submodularity can be tested with 
$\epsilon^{-\sqrt{n} \log n}$
queries and that it \emph{cannot} be tested with a number of queries that is independent of $n$. 
Subsequently, Feldman and \Vondrak~\cite{FV16} showed that in the $\ell_1$ testing framework, we can do much better: testing submodularity in this model requires a number of queries that is only logarithmic in $n$.
Our first result shows that, in fact, for any value of $p \ge 1$, it is possible to test submodularity in the $\ell_p$ setting with a number of queries that is completely \emph{independent} of $n$.

\begin{thm}
\label{thm:submodular}
For any $\epsilon > 0$ and any $p \ge 1$, there is an $\epsilon$-tester for submodularity in the $\ell_p$ testing model with query complexity 
$2^{\tilde{O}(1/\epsilon^{\max\{2,p\}})}$.
\end{thm}	

Another property that has been considered in the (standard Hamming distance) testing model is that of being a coverage function. Chakrabarty and Huang~\cite{CH15} showed that for constant values of $\epsilon > 0$, $O(nm)$ queries suffice to $\epsilon$-test whether a function $f$ is a coverage function on some universe $U$ of size $|U| \le m$. Note that, unlike in the learning and approximation settings, bounds on the number of queries required to test some property $\calP$ do not imply anything about number of queries required to test properties $\calP' \subset \calP$, so even though coverage functions are submodular, results on testing submodularity do not imply any bounds on the query complexity for testing coverage functions. Nonetheless, our next result shows that this property---along with most of the other properties of valuation functions listed above---can also be tested with a number of queries that is independent of $n$.

\begin{thm}
\label{thm:others}
For any $\epsilon > 0$ and any $p \ge 1$, there are $\epsilon$-testers in the $\ell_p$ testing model for 
\begin{itemize}
	\item additive functions, coverage functions, unit demand functions, OXS functions, and gross substitute functions that each have query complexity $2^{\tilde{O}(1/\epsilon^{\max\{2,p\}})}$; and
	\item fractional subadditivity and self-bounded functions that have query complexity $2^{2^{\tilde{O}(1/\epsilon^{\max\{2,p\}})}}$.
\end{itemize}
\end{thm}	

Theorems~\ref{thm:submodular} and~\ref{thm:others} are both special cases of a general testing result that we obtain by extending the technique of \emph{testing by implicit learning} of Diakonikolas et al.~\cite{DLMORSW07}.
We describe this general result in more details below.

\subsection{Testing real-valued functions by implicit learning}

There is a strong connection between property testing and learning theory that goes back to the seminal work of Goldreich, Goldwasser, and Ron~\cite{GGR98}. As they first observed, any proper learning algorithm for the class of functions that have some property $\mathcal{P}$ can also be used to test $\mathcal{P}$: run the learning algorithm, and verify whether the resulting hypothesis function $h$ is close to the tested function $f$ or not. 
This approach yields good bounds on the number of queries required to test many properties of functions, but, as simple information theory arguments show, it cannot yield query complexity bounds that are smaller than $\log n$ for almost all natural properties of functions over $\{0,1\}^n$. 

Diakonikolas et al.~\cite{DLMORSW07} bypassed this limitation for the special case when every function that has some property $\mathcal{P}$ is close to a junta. A function $f : \{0,1\}^n \to [0,1]$ is a \emph{$k$-junta} if there is a set $J \subseteq [n]$ of cardinality $|J| \le k$ such that the value of $f$ on any input $x$ is completely determined by the values $x_i$ for each $i \in J$. Every $k$-junta $f$ has corresponding ``core'' functions $f_{\mathrm{core}} : \{0,1\}^k \to \{0,1\}$ that define its value based on the value of the $k$ relevant coordinates of its input. Diakonikolas et al.'s key insight is that for testing properties whose functions are (very) close to juntas, it suffices to learn the core of the input function---without having to identify the identity of the relevant coordinates.

The wide applicability of the testing-by-implicit-learning methodology is due to the fact that for many natural properties of Boolean functions, the functions that have these properties must necessarily be close to juntas under the Hamming distance. The starting point for the current research is a recent breakthrough of Feldman and \Vondrak, who showed that a similar junta theorem holds for many properties of real-valued functions when closeness is measured according to $\ell_2$ distance.

\newtheorem*{FVthm}{Feldman--\Vondrak junta theorem}
\begin{FVthm}
Fix any $\epsilon \in (0, \frac{1}{2})$. For every $f : \{0,1\}^n \to [0,1]$,
\begin{itemize}
 	\item if $f$ is submodular then
 	there exists a submodular function $g: \{0, 1\}^n \to [0, 1]$ 
 	that is a $O(\frac{1}{\epsilon^2}\log \frac{1}{\epsilon})$-junta such that $\|f - g\|_2 \le \epsilon$; and
 	\item if $f$ is self-bounding then 
 	there exists a self-bounding function $g: \{0, 1\}^n \to [0, 1]$ 
 	that is a $2^{O(\frac{1}{\epsilon^2})}$-junta such that $\|f - g\|_2 \le \epsilon$;
 \end{itemize} 
\end{FVthm}

The logarithmic dependence on $n$ for the problem of testing submodularity in the $\ell_1$ testing model~\cite{FV16} follows directly from Feldman and \Vondrak's junta theorem and the (standard) testing-by-proper-learning connection. This junta theorem also suggests a natural approach for obtaining a constant query complexity for the same problem by combining it with a testing-by-implicit-learning algorithm. In order to implement this approach, however, new testing-by-implicit-learning techniques are required to overcome two obstacles.

The first obstacle is that all existing testing-by-implicit-learning algorithms~\cite{DLMORSW07,DLMSW08,CGM11,GOSSW11} are designed for properties that contain functions which are close to juntas in Hamming distance, not $\ell_p$ distance. This is a stronger condition, and enables the analysis of these algorithms to assume that with large probability, when $f$ is very close to a $k$-junta $f'$, the queries $x$ made by the algorithm all satisfy $f(x) = f'(x)$. In the $\ell_p$ distance model, however, we can have a function $f$ that is extremely close to a $k$-junta but still has $f(x) \neq f'(x)$ for many (or even every!) input $x$.

The second (related) obstacle that we encounter when considering submodular functions is that current testing-by-implicit-learning algorithms only work in the regime where the functions in $\calP$ are $\epsilon$-close to $k$-juntas for some $k < \epsilon^{-1/2}$. (See for example the discussion in \S2.5 of~\cite{Ser10}.) This condition is satisfied by the properties of Boolean functions that have been studied previously, but the bounds in the Feldman--\Vondrak junta theorem, however, do not satisfy this requirement.

We give a new algorithm for testing-by-implicit-learning that overcomes both of these obstacles. As a result, we obtain the following general theorem.

\begin{thm}
\label{thm:general}
For any $0 < \epsilon < \frac{1}{2}$ and any property $\calP$ of functions mapping $\{0,1\}^n \to [0,1]$, if $k \ge 1$ is such that for every function $f \in \calP$, there is a $k$-junta $h$ that satisfies $\|f - h\|_2 \le \frac{\epsilon}{10^6}$, then there is an $\epsilon$-tester for $\calP$ in the $\ell_2$ testing model with query complexity $\frac{2^{O(k\log k)}}{\epsilon^{10}}$.
\end{thm}

Theorems~\ref{thm:submodular} and~\ref{thm:others} are both obtained directly from Theorem~\ref{thm:general}, the Feldman--\Vondrak junta theorem and Fact~\ref{fact:lp-testing}.

\subsection{Overview of the proofs}

\paragraph{The algorithm.}
The current testing-by-implicit-learning algorithms proceed in two main stages. In the first stage, the coordinates in $[n]$ are randomly partitioned into $\poly(k)$ parts, and an influence test is used to identify the (at most $k$) parts that contain relevant variables of an unknown input function $f$ that is very close to being a $k$-junta. In the second stage, inputs $x \in [n]$ are drawn at random according to some distribution, the value $f(x)$ is observed, and the value of the relevant coordinate in each of the parts identified in the second stage is determined using more calls to the influence test.

The \textsc{Implicit Learning Tester} algorithm that we introduce in this paper reverses the order of the two main stages. In the first stage, it draws a sequence of $q$ queries $X = (x^{(1)},\ldots,x^{(q)})$ at random and queries the value of $f$ on each of these queries. It also uses $X$ to partition the coordinates in $[n]$ into $2^q$ random parts according to the values of the coordinate on the $q$ queries. In the second stage, the algorithm then uses an influence estimator to identify the $k$ parts that contain the relevant coordinates of a $k$-junta that is close to $f$ and, since all the coordinates in a common part have the same value on each of the $q$ queries, learn the value of the $k$ relevant coordinates on each of these initial queries. The algorithm then checks whether the core function thus learned is consistent with those of functions in the property being tested.

The main advantage of the \textsc{Implicit Learning Tester} algorithm that its analysis does not require the assumption that our samples are exactly consistent with those of an actual $k$-junta (instead of those of a function that is only promised to be \emph{close} to a $k$-junta). This feature enables us to overcome the obstacles listed in the previous section, at the cost of adding a few complications to the analysis, as described below.

\paragraph{The analysis.}
There are two main technical ingredients in the analysis of the algorithm. The first, established in Lemma~\ref{lem:Influence of variables outside the selected buckets is low}, is used to show that when $f$ is close to a $k$-junta in $\ell_2$ distance, the search procedure identifies parts that contain the $k$ relevant coordinates of some $k$-junta that is close to $f$. (Note that the search is not guaranteed to find the parts that contain the relevant coordinates of the $k$-junta that is \emph{closest} to $f$, but it suffices to find those of any close $k$-junta.)

The second technical ingredient addresses the fact that by drawing the $q$ samples $x^{(1)},\ldots,x^{(q)}$ first and then using these samples to provide the initial partition of the coordinates in $[n]$, we no longer will obtain uniformly random samples of the core $f_{\core}$ of the input function $f$. Nonetheless, in Lemma~\ref{lem:main-dist-estimate}, we show that when $f$ is close to a $k$-junta, the distribution of these samples on the core function still enables us to accurately estimate the distance of $f_{\core}$ to the core functions of any other $k$-junta.

\subsection{Discussion and open problems}

Theorems~\ref{thm:submodular}--\ref{thm:general} raise a number of intriguing questions. The most obvious question left open is whether we can also test subadditivity of real-valued functions with a constant number of queries: subadditive functions need not be close to juntas, so such a result would appear to require a different technique.

It is also useful to compare our bounds for submodularity testing with those for testing monotonicity: in the Hamming distance testing model, Seshadhri and \Vondrak~\cite{SV11} showed that the query complexity for testing submodularity is at least as large as that for testing monotonicity. However, the best current bounds for testing monotonicity in the $\ell_p$ testing model have a linear dependence on $n$~\cite{BRY14}. Is it also possible to test monotonicity with a constant number of queries? Or is it the case that testing submodularity is strictly easier than testing monotonicity in the $\ell_p$ testing setting?

\section{Preliminaries}
\label{sec:prelims}

Let $\calF_n$ denote the set of functions mapping $\{0,1\}^n$ to $[0,1]$. 
For any $f \in \calF_n$ and $S \subseteq [n]$ with complement $\overline{S} = [n] \setminus S$, when 
$x \in \{0, 1\}^{S}$ and $y \in \{0, 1\}^{\overline{S}}$, we write $f(x,y)$ to denote the value $f(z)$ for the input $z$ that satisfies $z_i = x_i$ for each $i \in S$ and $z_i = y_i$ otherwise.

We use the standard definitions and notation for the Fourier analysis of functions $f : \{0,1\}^n \to [0,1]$. For a complete introduction to the topic, see~\cite{ODo14}. Throughout the paper, unless otherwise specified all probabilities and expectations are over the uniform distribution on the random variable's domain.

\subsection{Property testing}

A \emph{property} $\calP$ of functions in $\calF_n$ is a subset of these functions that is invariant under relabeling of the $n$ coordinates.
The \emph{Hamming distance} between $f,g \in \calF_n$ is $\distH(f,g) = \Pr_x[ f(x) \neq g(x)]$ and the Hamming distance between $f$ and a property $\calP$ is $\distH(f,\calP) = \inf_{g \in \calP} \distH(f,g)$.
For $p \ge 1$, the \emph{$\ell_p$ distance} between $f$ and $g$ is $\distp(f,g) = \|f - g\|_p = \left( \E_x[ |f(x) - g(x)|^p ] \right)^{1/p}$ and the $\ell_p$ distance between $f$ and $\calP$ is $\distp(f,\calP) = \inf_{g \in \calP} \distp(f,g)$.

Given $\epsilon > 0$, An \emph{$\epsilon$-tester} in the Hamming testing model (resp., $\ell_p$ testing model) for some property $\calP \subseteq \calF_n$ is a randomized algorithm that (i) accepts every function $f \in \calP$ with probability at least $\frac23$ and (ii) rejects every function $f$ that satisfies $\distH(f,\calP) \ge \epsilon$ (resp., $\dist_p(f,\calP) \ge \epsilon$) with probability at least $\frac23$. An $\epsilon$-tester for $\calP$ is an \emph{$(\epsilon',\epsilon)$-tolerant tester}, for some $\epsilon' < \epsilon$ if it additionally accepts every function $f$ that satisfies
$\distH(f,\calP) \le \epsilon'$ (resp., $\dist_p(f,\calP) \le \epsilon$) with probability at least $\frac23$.

Our proofs of Theorems~\ref{thm:submodular}--\ref{thm:general} are established in the $\ell_2$ testing model. The result for general $\ell_p$ testing models is obtained from the following elementary relation between the query complexities of testing any property in different $\ell_p$ testing models.

\begin{fact}[c.f.~Fact 5.2 in~\cite{BRY14}]
\label{fact:lp-testing}
	For any $\calP \subseteq \calF_n$, any $\epsilon > 0$, and any $p \ge 1$,
	the number $Q_p(\calP,\epsilon)$ of queries required to $\epsilon$-test $\calP$ in the $\ell_p$ testing model satisfies
	$
	Q_2(\mathcal{P}, \epsilon) 
	\le Q_p(\mathcal{P}, \epsilon) 
	\le Q_2(\mathcal{P}, \epsilon^{\frac{p}{2}}).
	$
\end{fact}

Theorem~\ref{thm:others} also relies on the following hierarchy of properties. (See, e.g.,~\cite{LLN06}.)

\begin{lemma}
\label{lem:hierarchy}
The properties of $\calF_n$ defined in the introduction satisfy the inclusion hierarchy
\begin{align*}
  \textrm{Additive} \subseteq \textrm{Coverage} \subseteq \textrm{Unit~demand}
  \subseteq \textrm{OXS} \subseteq \textrm{Gross~substitute}
  \subseteq \textrm{Submodularity} \subseteq \textrm{XOS} \subseteq \textrm{Self-bounding}. 
\end{align*}
\end{lemma}

\subsection{Juntas}

The function $f: \{0, 1\}^n \to [0, 1]$ is a \emph{junta on the set $J \subseteq [n]$} if for every $x,y \in \{0,1\}^n$ that satisfy $x_i = y_i$ for every $i \in J$, we have $f(x) = f(y)$. The function $f$ is a \emph{$k$-junta} if it is a junta on some set $J \subseteq [n]$ of cardinality $|J| \le k$.
The function $f_{\rm core} : \{0,1\}^k \to [0,1]$ is a \emph{core function} of the $k$-junta $f : \{0,1\}^n \to [0,1]$ if there is a projection $\psi : \{0,1\}^n \to \{0,1\}^k$ defined by setting $\psi(x) = (x_{i_1},\ldots,x_{i_k})$ for some distinct $i_1,\ldots,i_k \in [n]$ such that for every $x \in \{0,1\}^n$, $f(x) = f_{\rm core}\big(\psi(x)\big)$.

\begin{defn}	
For any function $f: \{0, 1\}^n \to [0, 1]$ and set $J \subseteq [n]$, the \emph{$J$-junta projection} of $f$ is the function $f_J: \{0, 1 \}^{J} \to [0, 1]$ defined by setting $f_J(x) = \E_{y \in \{0,1\}^{\overline{J}}}[f(x, y)]$ for every $x \in \{0,1\}^J$.
\end{defn}

A basic fact that we will require is that $f_J$ is the $J$-junta that is closest to $f$ under the $\ell_2$ metric.

\begin{prop}
\label{prop:f_S is the closest k-junta}
	For every $f : \{0,1\}^n \to [0,1]$ and $J \subseteq [n]$, if $g : \{0,1\}^n \to [0,1]$ is a $J$-junta, then $\disttwo(f,f_J) \le \disttwo(f,g)$.
\end{prop}	
	
\begin{proof}
	By applying the identity $\|f - g\|_2^2 = \|f - f_J + f_J - g\|_2^2$ and by expanding the right-hand side, we obtain
	\begin{align*}
		\|f - g\|_2^2 
		&= \E_{x\in \{0, 1\}^{J}}\left[ \E_{y\in \{0, 1\}^{\overline{J}}}\Big[\big(f(x,y) - f_J(x,y) + f_J(x,y) - g(x,y)\big)^2\Big]\right]\\
		&= \|f - f_J\|_2^2 + \|f_J - g\|_2^2 + 2\E_{x}\Big[ \E_{y}\Big[\big(f(x,y) - f_J(x,y)\big)\big(f_J(x,y)-g(x,y)\big)\Big]\Big].
	\end{align*}
	Since $f_J - g$ is a $J$-junta, it does not depend on $y$ and, by the definition of $f_J$, the last term equals $0$. Therefore, $\|f - g\|_2^2 = \|f - f_J\|_2^2 + \|f_J - g\|_2^2$ and the claim follows.
\end{proof}

The property $\calP \subseteq \calF_n$ is a \emph{property of $k$-juntas} if every function $f \in \calP$ is a $k$-junta. The \emph{core property} of a property $\calP$ of $k$-juntas is the property $\calP_{\rm core} \subseteq \calF_k$ defined by $\calP_{\rm core} = \{ f_{\rm core} : f \in \calP \}$.
For any $\gamma > 0$, the \emph{$\gamma$-discretized approximation} of a function $f \in \calF_n$ is the function $f^{(\gamma)}$ obtained by rounding the value $f(x)$ for each $x \in \{0,1\}^n$ to the nearest multiple of $\gamma$. The $\gamma$-discretized approximation of a property $\calP$ is the property $\calP^{(\gamma)} = \{f^{(\gamma)} : f \in \calP\}$.

\subsection{Influence}

The notion of influence of coordinates in functions over the Boolean hypercube plays a central role in both our algorithm and its analysis. Informally, the influence of a set of coordinate measures how much re-randomizing these coordinates affects the value of the function. This notion is made precise as follows.

\begin{defn}
	The \emph{influence} of a set $S \subseteq [n]$ of coordinates in the function $f : \{0,1\}^n \to [0,1]$ is 
	\[
	\Inf_f(S) := \E_{x \in \{0,1\}^{\overline{S}}}\big[ \Var_{y \in \{0,1\}^S} f(x,y) \big]
	= \tfrac12 \E_{x \in \{0,1\}^{\overline{S}}}\Big[ \E_{y,y' \in \{0,1\}^S}\Big[ \big( f(x,y) - f(x,y') \big)^2 \Big] \Big].
	\]
\end{defn}

Our proofs make use of a few standard facts regarding the influence of sets of coordinates in $f$.

\begin{fact}
\label{fact:inf-fourier}
	The influence of $S \subseteq [n]$ in $f \in \calF_n$ is
	$
	\Inf_f(S) = \sum_{T: T\cap S \neq \emptyset}\hat{f}^2(T).
	$
\end{fact}

\begin{fact}
\label{fact:inf-monotone}
	For every $f \in \calF_n$ and $S, T \subseteq [n]$, we have
	$
	\Inf_f(S) \le \Inf_f(S \cup T) \le \Inf_f(S) + \Inf_f(T).
	$
\end{fact}	

\begin{fact}
\label{fact:inf-junta}
For every $f \in \calF_n$ and $J \subseteq [n]$, we have
$
\Inf_f(\overline{J}) = \disttwo(f, f_J)^2.
$
\end{fact}

\begin{prop}
\label{prop:Relation_between_distance_and_influence}
	Fix $\epsilon > 0$, and let $f, g: \{0, 1\}^n \to [0, 1]$ satisfy $\disttwo(f, g) \leq \epsilon$. Then for any set $S \subseteq [n]$, 
	$
	|\Inf_f(S)^{\frac12} - \Inf_g(S)^{\frac12}| \leq \epsilon.
	$
\end{prop}

\begin{proof} 
By Fact~\ref{fact:inf-junta}, we have $\Inf_f(S)^{\frac12} = \|f - f_{\overline{S}}\|_2$ and $\Inf_g(S)^{\frac12} = \|g - g_{\overline{S}}\|_2$. By Proposition~\ref{prop:f_S is the closest k-junta}, we also have that $\|f - f_{\overline{S}}\|_2 \le \|f - g_{\overline{S}}\|_2$. Combining these observations with the triangle inequality, we obtain
$
    \Inf_f(S)^{\frac12} - \Inf_g(S)^{\frac12} 
    = \|f - f_{\overline{S}}\|_2 - \|g - g_{\overline{S}}\|_2 
    \le \|f - g_{\overline{S}}\|_2 - \|g - g_{\overline{S}}\|_2 
    \le \|f - g\|_2
    \le \epsilon.
$
Hence $\Inf_f(S)^{\frac12} - \Inf_g(S)^{\frac12} \leq \epsilon$ and, similarly, $\Inf_g(S)^{\frac12} - \Inf_f(S)^{\frac12} \le \epsilon$ as well.
\end{proof}	

\begin{prop}
\label{prop:infest}
    There is an algorithm \EstInf such that for every $f: \{0, 1\}^n \rightarrow [0, 1]$, $S \subseteq [n]$, $m \ge 1$, and $t \ge 0$, it makes $m$ queries to $f$ and returns an estimate of the influence of $S$ in $f$ that satisfies
	\[
	\Pr\big[|\Inf_f(S)  - \EstInf(f,S,m)| \geq t\big] \leq  2 e^{-2mt^2}.
	\]
\end{prop}

We also use the following key lemma from~\cite{BWY15}.

\begin{lemma}[Lemma 2.3 in \cite{BWY15}]
\label{lem:infparts}
Let $f: \{0, 1\}^n \rightarrow [0, 1]$ be a function that is $\epsilon$-far from $k$-juntas and $P$ be a random partition of $[n]$ into $r > 20 k^2$ parts. Then with probability at least $\frac{5}{6}$, $\Inf_f(\overline{J}) \geq \frac{\varepsilon^2}{4}$ for any union $J$ of $k$ parts from $P$.
\end{lemma}
For the reader's convenience, we include the proof of Lemma~\ref{lem:infparts} in Appendix~\ref{app:proofs}; though the original lemma in~\cite{BWY15} was only for Boolean-valued functions, the proof remains essentially unchanged.

\section{Testing by implicit learning}
	
The proof of Theorem~\ref{thm:general} is established by analyzing the \textsc{Implicit Learning Tester} algorithm.
		
\begin{algorithm}[t]
	\KwData{$q = \frac{2^{O(k)}}{\epsilon^5}$, $m = O(\frac{k^6}{\epsilon^5})$, $r = \log \frac{2^k}{100k^4}$}
	Draw $x^{(1)},\ldots, x^{(q)} \in \{0, 1\}^n$ independently and uniformly at random\;
	For each $c \in \{0,1\}^q$, define $S_c \gets \big\{i \in [n] : \big(x^{(1)}_i, \ldots, x^{(q)}_i\big) = c\big\}$\;
	\medskip
	
	Let $P_1,\ldots,P_{100k^4}$ be a random equi-partition of $\{0, 1\}^q$\;
	\For{each $J \subseteq [100k^4]$ of size $|J| = k$}{
		$S_J \gets \bigcup_{j \in J} \bigcup_{c \in P_j} S_c$\;
		$\eta_J \gets \EstInf(f, [n]\setminus S_J, m)$\;
	}
		
	$\{j^*_1,\ldots,j^*_k\} \gets \argmin_J \eta_J$\;
	$(P_{0,1},\ldots,P_{0,k}) \gets (P_{j^*_1}, \ldots, P_{j^*_k})$\;

	\medskip

	\For{$\ell = 1,\ldots,r$}{
		Let $P_{\ell,i,0}, P_{\ell,i,1}$ be a random equi-partition of $P_{\ell-1,i}$ for each $i \le k$\;
		
		\For {every $z \in \{0,1\}^k$}{
			$S_z \gets \bigcup_{i \le k} \bigcup_{c \in P_{\ell,i,z_i}} S_c$\;
			$\eta_z \gets \EstInf(f, [n]\setminus S_z, m)$\;
		}
	
		$z^*_\ell \gets \argmin_z \eta_z$\;
		For each $i \le k$, update $P_{\ell,i} \gets P_{\ell,i,z^*_\ell}$\;
	}
	\medskip

	Let $B = \{b_1,..., b_k\} \gets \bigcup_{i \le k} P_{r,i}$\; 
	
	If $\EstInf(f, [n]\setminus S_B, m) > \epsilon^2/1000$, reject\;
	
	Let $\phi: \{0, 1\}^n \rightarrow \{0, 1\}^k$ be any projection that satisfies $\phi(x)_i \in S_{b_i}$ for each $i \le k$\;
	
	\medskip
	\For{ $h \in \calF_{\core}^{(\frac{\epsilon}{1000})}$} {
		If $\frac1q \sum_{i=1}^{q} \big(f(x^{(i)}) - h(\phi(x^{(i)}))\big)^2 \leq 0.35 \epsilon$, accept and return $h$\; 
	}
	Reject\;
	
	\caption{\textsc{Implicit Learning Tester}$(\calF,k,\epsilon)$}
	\label{alg:implicit-learning}
\end{algorithm}

\subsection{Proof of Theorem~\ref{thm:general}}	

The analysis of the \textsc{Implicit Learning Tester} relies on two technical lemmas. The first shows that when the input function $f$ is 
close to a $k$-junta, then with reasonably large probability, the function $f$ is 
close to a junta on the set $B$ of $k$ parts that is identified by the algorithm. 

\begin{lemma}
\label{lem:Influence of variables outside the selected buckets is low}
	For any $\varepsilon > 0$, if the function $f : \{0,1\}^n \to [0,1]$ is $\varepsilon$-close
	to a $k$-junta and every call to \EstInf returns an influence estimate with additive error at most $\frac{\varepsilon^2}{100 k^2}$, then the set $B$ obtained by the \textsc{Junta-Property Tester} satisfies
	$
	\Pr\big[\Inf_f([n] \setminus S_B) > 100\varepsilon^2\big] \leq \tfrac{1}{20}.
	$
\end{lemma}

The second lemma shows that the estimate in Step 20 provides a good estimate of the distance between $f$ and the functions in $\calP$.

\begin{lemma}
\label{lem:main-dist-estimate}
Fix $\varepsilon > 0$. Let $f : \{0,1\}^n \to [0,1]$ be a function that satisfies $\dist_2(f,g) \le \varepsilon$ for some function $g$ that is a junta on $J \subseteq [n]$, $|J| \le k$. Then for every $h_{\core} \in \calF_{\core}^{(\frac{\varepsilon}{1000})}$, the mapping $\psi : \{0,1\}^n \to \{0,1\}^k$ defined in the \textsc{Implicit Learning Tester} and the function $h = h_{\core} \circ \psi$ satisfy
\[
\Big| \Big( \tfrac 1q \sum_{i}^q \big(f(x^{(i)})) - h(x^{(i)})\big)^2 \Big)^{\frac{1}{2}} - \dist_2(g,h) \Big| \le  3\varepsilon
\]
except with probability at most $2e^{-16q\varepsilon^4} + \frac{5k^2}{2^q}$.
\end{lemma}

The proofs of these lemmas are presented in Sections~\ref{sec:first-lemma} and~\ref{sec:second-lemma}. We now show how they are used to complete the proof of Theorem~\ref{thm:general}.

As a first observation, we note that by Hoeffding's inequality and the union bound, all of the calls to \textsc{EstimateInf} have additive error at most $\frac{\epsilon^2}{10^6k^2}$ except with probability at most $\frac16$. In the following, we assume that this condition holds and show how, when it does, the algorithm correctly accepts or rejects with probability with probability at least $\frac56$.

\begin{claim}[Completeness]
\label{claim:completeness}
When $f$ is $\frac{\epsilon}{10^6}$-close to the property $\calF$ of $k$-juntas, the \textsc{Implicit Learning Tester} accepts with probability at least $\frac56$.
\end{claim}

\begin{proof}
First, by Lemma~\ref{lem:Influence of variables outside the selected buckets is low}, the probability that $f$ is rejected on step 17 is at most $\frac1{18}$. In the rest of the proof, we will show that except with probability at most $\frac19$, there is a function $h_{\core} \in \calF_{\core}^{(\frac{\epsilon}{1000})}$ for which the algorithm accepts on line 20.

Let $g \in \calF$ be a function that satisfies $\dist_2(f,g) \le \frac{\epsilon}{10^6}$. Without loss of generality, we can assume that $g$ is a junta on $[k]$.
Let $J = [k] \cap S_B$ be the set of the junta variables of $g$ that are contained in the final parts selected by the algorithm. Again without loss of generality (by relabeling the input variables once again if necessary), we can assume that $J = [j]$ for some $j \le k$, and $i \in S_{b_i}$, for $i \le j$.

Define $\psi : \{0,1\}^n \to \{0,1\}^k$ to be the mapping defined by $\psi(x) = (x_1,\ldots,x_j,x_{i_1},\ldots,x_{i_{k-j}})$ where $i_1,\ldots,i_{k-j} \in [n] \setminus [k]$ are representative coordinates from the remaining parts $b \in B$ for which $P_b \cap [k] = \emptyset$.

Let $g_{\core} \in \calF_{\core}$ be the core of $g$ corresponding to the projection $\psi(x) = (x_1,\ldots, x_k )$, and let $h_{\core} \in \calF_{\core}^{(\frac\epsilon{10^6})}$ be the discretized approximation to $g_{\core}$. Define $h = h_{\core} \circ \psi$. 
By our choice of $g$, we have $\dist_2(f,g) \le \frac{\epsilon}{10^6}$.
In order to invoke Lemma~\ref{lem:main-dist-estimate}, we now want to bound $\dist_2(g,h)$.

Let $h^* \in \calF^{(\frac{\epsilon}{10^6})}$, be the discretized approximation of $g$. Then $\dist_2(g,h^*) \le \frac{\epsilon}{10^6}$ and the triangle inequality implies that
\[
\dist_2(f,h^*) \le \dist_2(f,g) + \dist_2(g,h^*) \le \tfrac{2\epsilon}{10^6}
\]
and that
\[
\dist_2(g,h) \le \dist_2(g,h^*) + \dist_2(h^*,h) \le \dist_2(h^*,h) + \tfrac{\epsilon}{10^6}.
\]
Furthermore, since $h_{\core} = h^*_{\core}$,
\begin{align*}
\dist_2(h^*,h) &= \E_x\Big[ \big( h^*_{\core}(x_1,\ldots,x_k) - h^*_{\core}(x_1,\ldots,x_j,x_{i_1},\ldots,x_{i_{k-j}})\big)^2]^{\frac12} \\
&= 2\,\Inf_{h^*_{\core}}( [k] \setminus [j] )^{\frac12} 
= 2\,\Inf_{h^*}( [n] \setminus [j])^{\frac12}.
\end{align*}
By Proposition~\ref{prop:Relation_between_distance_and_influence} and Lemma~\ref{lem:Influence of variables outside the selected buckets is low}, except with probability at most $\frac1{18}$, 
\[
\Inf_{h^*}( [n] \setminus [j])^{\frac12} \le
\Inf_f([n] \setminus [j])^{\frac12} + \dist_2(f,h^*)
\le \Inf_f([n] \setminus S_B)^{\frac12} + \tfrac{2\epsilon}{10^6} \le
\tfrac{12\epsilon}{10^6}
\]
and the distance between $g$ and $h$ is bounded by $\dist_2(g,h) \le \frac{13}{10^6}\epsilon$. When this bound holds, by Lemma~\ref{lem:main-dist-estimate} with $\varepsilon = \frac{\epsilon}{100}$, the algorithm accepts $f$ for this $h$ except with probability at most $\frac{1}{18}$.
\end{proof}

\begin{claim}[Soundness I]
\label{claim:soundnessI}
If $f$ is $\frac{\epsilon}{100}$-far from being a $k$-junta, then the \textsc{Implicit Learning Tester} rejects with probability at least $\frac56$.
\end{claim}

\begin{proof}
The initial partition $S_{P_1},\ldots,S_{P_{100k^4}}$ is a random partition of $[n]$ with more than $20k^2$ parts so, by Lemma~\ref{lem:infparts}, with probability at least $\frac{5}{6}$, for any union $J \subseteq [n]$ of at most $k$ of these parts we have $\Inf_f([n] \setminus J) \ge \frac{\epsilon^2}{400}$. When this is the case, the inclusion $S_B \subseteq \overline{L_0}$ and the fact that $L_0$ is the complement of the union of some set of $k$ parts in the random partition imply that
\[
\Inf_f([n] \setminus S_B) \ge \Inf_f(L_0) \ge \frac{\epsilon^2}{400}
\]
and, under the assumed accuracy of \textsc{EstimateInf} calls, the algorithm rejects $f$ in Step 17.
\end{proof}

\begin{claim}[Soundness II]
\label{claim:soundnessII}
If $f$ is $\frac{\epsilon}{100}$-close to a $k$-junta, but is $\frac{99\epsilon}{100}$-far from $\calF$, then the \textsc{Implicit Learning Tester} rejects with probability at least $\frac56$.
\end{claim}

\begin{proof}	
Let $g$ be any $k$-junta that satisfies $\dist_2(f,g) \le \frac{\epsilon}{100}$. For any $h_{\core} \in \calF_{\core}^{(\frac{\epsilon}{1000})}$ and any injective mapping $\psi : \{0,1\}^n \to \{0,1\}^k$, the function $h = h_{\core} \circ \psi$ is in $\calF^{(\frac{\epsilon}{1000})}$ and so by the triangle inequality,
\[
dist_2(f, \calF^{(\frac{\epsilon}{1000})}) \geq dist_2(f, \calF) - \tfrac{\epsilon}{1000}
\]
and
\[
\dist_2(g,h) \ge \dist_2(f,h) - \dist_2(f,g) \ge \tfrac{99}{100}\epsilon -  \tfrac{\epsilon}{1000} -  \tfrac{\epsilon}{100} \ge \tfrac{97}{100}\epsilon.
\]
Then, by Proposition~\ref{prop:Relation_between_distance_and_influence} and the union bound over all $|\calF_{\core}^{(\frac{\epsilon}{1000})}| \le (1000/\epsilon)^{2^k}$ functions in $\calF_{\core}^{(\frac{\epsilon}{1000})}$, with probability at least $\frac56$, the condition in Step 20 is never satisfied and the algorithm rejects. 
\end{proof}

To complete the proof of Theorem~\ref{thm:general} in the case where $p=2$, consider now any property 
$\calP$ that contains only functions which are $\frac{\epsilon}{10^6}$-close to some $k$-junta. Let $\calF$ be the property that includes all $k$-juntas that are $\frac{\epsilon}{10^6}$-close to $\calP$.
Claim~\ref{claim:completeness} shows that \textsc{Implicit Learning Tester} accepts every function in $\calP$ with the desired probability, and Claims~\ref{claim:soundnessI} and~\ref{claim:soundnessII} shows that it rejects all functions that are $\epsilon$-far from $\calP$.
Finally, we note that the query complexity of the algorithm is at most $q + 2m( 2^{O(k\log(k))} + 2^kq) = \frac{2^{O(k \log k)}}{\epsilon^{10}}$, as claimed. Finally, the general result for $\ell_p$ testing when $p\neq 2$ follows from Fact~\ref{fact:lp-testing}.

\subsection{Proof of Lemma~\ref{lem:Influence of variables outside the selected buckets is low}}
\label{sec:first-lemma}

	Let $f$ be any function $\varepsilon$-close to a $k$-junta and assume without loss of generality (by relabeling the input variables if necessary) that $f$ is close to a junta on $[k]$. The definition of $P_1,\ldots,P_{100k^4}$ in step 3, means that $S_{P_1},\ldots,S_{P_{100k^4}}$ is a random partition of $[n]$. So by the union bound, the probability that any two of the coordinates in $[k]$ land in the same part is at most $\frac1{100k^2}$. 
	
	For each $\ell = 0,1,2,\ldots,r$, let $L_\ell = [n] \setminus \bigcup_{i=1}^k S_{P_{\ell,i}}$ denote the set of variables that have been ``eliminated'' after $\ell$ iterations of the loop. Then $[n] \setminus S_B = L_r$ and
	\begin{align}
	\label{eqn:telescope}
	\Inf_f([n] \setminus S_B) = \Inf_f(L_0) + \sum_{\ell=1}^r \Big( \Inf_f(L_\ell) - \Inf_f(L_{\ell-1}) \Big).
	\end{align}
	We bound both terms on the right-hand side of the expression separately.
	
	By Proposition~\ref{prop:Relation_between_distance_and_influence}, we have $\Inf_f([n] \setminus [k]) \le \varepsilon^2$ and so by the monotonicity of influence there is a choice of $J \subseteq [k^2]$ of size $|J| \le k$ for which 
	$\Inf_f([n] \setminus S_J) \le \varepsilon^2$. The guaranteed accuracy on \EstInf then implies that
	\begin{equation}
	\Inf_f(L_0) \le (1 + \tfrac2{100k^2}) \varepsilon^2.
	\end{equation}
	
	Define $\mathcal{E} = \{\ell \le r : (L_\ell \setminus L_{\ell-1}) \cap [k] \neq \emptyset\}$ to be the set of rounds for which the algorithm eliminated at least one of the coordinates in $[k]$. 
	By this definition, each $\ell \in [r] \setminus \mathcal{E}$ satisfies $(L_\ell \setminus L_{\ell-1}) \cap [k] = \emptyset$ and
	\begin{align}
	\sum_{\ell \in [r] \setminus \mathcal{E}} \Inf_f(L_\ell) - \Inf_f(L_{\ell-1})
	&= \sum_{\ell \in [r] \setminus \mathcal{E}} ~\sum_{T:\,T \cap L_\ell \neq \emptyset \wedge T \cap L_{\ell-1} = \emptyset} \hat{f}(T)^2 \nonumber \\
	&\le \sum_{T \subseteq [n] \setminus [k]} \hat{f}(T)^2 \le \Inf_f([n] \setminus [k]) \le \varepsilon^2.
	\end{align} 
	For each $\ell \in \mathcal{E}$, define $X_\ell = \{\cup_{i=1}^{k} S_{P_{\ell,i, 1 - (z^*_\ell)_i}}: S_{P_{\ell,i, 1 - (z^*_\ell)_i}} \cap [k] \neq \emptyset \}$ to be the set of coordinates in the parts that contain a coordinate in $[k]$ that was eliminated in the $\ell$th iteration of the loop. Let also $Y_\ell = \{\cup_{i=1}^{k} S_{P_{\ell,i, (z^*_\ell)_i}}: S_{P_{\ell,i, 1 - (z^*_\ell)_i}} \cap [k] \neq \emptyset \}$ be the coordinates in the parts that were kept instead. Then the guaranteed accuracy of \EstInf and the choice of $z^*_\ell$ implies that
	\[
	\Inf_f( L_\ell ) \le \Inf_f\big( (L_\ell \setminus X_\ell) \cup Y_\ell \big) + 2\,\tfrac{\varepsilon^2}{100k^2}
	\]
	and, therefore,
	\begin{align}
	\sum_{\ell \in \mathcal{E}} \Inf_f(L_\ell)& - \Inf_f(L_{\ell-1})
	\le \frac{2\varepsilon^2}{1000k} + \sum_{\ell \in \mathcal{E}}
	\Inf_f\big( (L_\ell \setminus X_\ell) \cup Y_\ell \big) - \Inf_f( L_{\ell-1} ) \nonumber \\
	&\le \frac{2\varepsilon^2}{1000k} + 
	\sum_{\ell \in \mathcal{E}} \, \sum_{T:\,T \cap (L_\ell \setminus X_\ell) \neq \emptyset \wedge T \cap L_{\ell-1} = \emptyset}
	\hat{f}(T)^2
	+ 
	\sum_{\ell \in \mathcal{E}} \, \sum_{T:\,T \cap Y_\ell \neq \emptyset \wedge T \cap L_{\ell-1} = \emptyset}
	\hat{f}(T)^2.
	\label{eqn:splitXY}
	\end{align}
	As above, since $(L_\ell \setminus X_\ell) \cap [k] = \emptyset$, 
	\begin{equation}
	\sum_{\ell \in \mathcal{E}} \, \sum_{T:\,T \cap (L_\ell \setminus X_\ell) \neq \emptyset \wedge T \cap L_{\ell-1} = \emptyset}
	\hat{f}(T)^2 \le \sum_{T \subseteq [n] \setminus [k]} \hat{f}(T)^2 \le \varepsilon^2.
	\end{equation}
	It remains to bound the last sum on the right-hand side of~\eqref{eqn:splitXY}. By splitting up the terms in this sum according to whether $|T| \le k$ or not, we obtain
	\[
	\sum_{T: T \cap Y_\ell \neq \emptyset \wedge T \cap L_{\ell-1} = \emptyset} \hat{f}(T)^2 
	\le 
	\sum_{|T| \le k} \hat{f}(T)^2 \cdot \mathbf{1}[T \cap Y_\ell \neq \emptyset] 
	+ \sum_{|T| > k} \hat{f}(T)^2 \cdot \mathbf{1}[T \cap L_{\ell-1} = \emptyset].
	\]
	Let $Z \subseteq [n] \setminus [k]$ denote the set of coordinates that occur in one of the the original parts $S_{P_1},\ldots,S_{P_{100k^4}}$ that also contains one of the elements in $[k]$. Then $Y_\ell \subseteq Z$ and
	\[
	\sum_{\ell \in \mathcal{E}} 
	\sum_{|T| \le k} \hat{f}(T)^2 \cdot \mathbf{1}[T \cap Y_\ell \neq \emptyset]
	\le
	\sum_{\ell \in \mathcal{E}} 
	\sum_{|T| \le k} \hat{f}(T)^2 \cdot \mathbf{1}[T \cap Z \neq \emptyset]
	\le k \cdot
	\sum_{|T| \le k} \hat{f}(T)^2 \cdot \mathbf{1}[T \cap Z \neq \emptyset].
	\]
	The probability, over the choice of $P_1,\ldots,P_{100k^4}$, that $T \cap Z \neq \emptyset$ is at most $|T|/100k^3$, so the expected value of the last expression (again over the choice of the initial partition) is bounded above by
	\begin{equation}
	\E \Big[ \sum_{\ell \in \mathcal{E}} 
	\sum_{|T| \le k} \hat{f}(T)^2 \cdot \mathbf{1}[T \cap Y_\ell \neq \emptyset] \Big]
	\le k \cdot \sum_{|T| \le k, T \setminus [k] \neq \emptyset} \hat{f}(T)^2 \cdot \Big(\frac{k}{100k^3}\Big)
	\le \frac{1}{100k} \cdot \Inf_f([n] \setminus [k]) \le \frac{\varepsilon^2}{100k}.
	\end{equation}
	Lastly, since $L_0 \subseteq L_{\ell-1}$ for each $\ell \ge 1$, 
	\[
	\sum_{|T| > k} \hat{f}(T)^2 \cdot \mathbf{1}[T \cap L_{\ell-1} = \emptyset]
	\le 
	\sum_{|T| > k} \hat{f}(T)^2 \cdot \mathbf{1}[T \cap L_0 = \emptyset].
	\]
	A set $T$ can be disjoint from $L_0$ only when its elements are contained in at most $k$ of the parts of the initial random partition, which happens with probability at most $\frac1{100k^2}$ when $|T| > k$, so 
	\begin{equation}
	\label{eqn:mainlem-last}
	\E \Big[ 
	\sum_{\ell \in \mathcal{E}} \sum_{|T| > k} \hat{f}(T)^2 \cdot \mathbf{1}[T \cap L_{\ell-1} = \emptyset] \Big] 
	\le \E \Big[ k \sum_{|T| > k} \hat{f}(T)^2 \cdot \mathbf{1}[T \cap L_0 = \emptyset] \Big] \le 
	\frac{1}{100k} \sum_{|T| > k} \hat{f}(T)^2 
	\le \frac{\varepsilon^2}{100k},
	\end{equation}
	where the last inequality uses the fact that $\sum_{|T| > k} \hat{f}(T)^2 \le \Inf_f([n] \setminus [k])$.
	
	Combining the inequalities~\eqref{eqn:telescope}--\eqref{eqn:mainlem-last}, we obtain that the expected value of $\Inf_f([n] \setminus S_B)$ is bounded above by
	\[
	\E\big[ \Inf_f([n] \setminus S_B) \big] 
	\le (1 + \tfrac2{100k^2}) \varepsilon^2 + \tfrac{2\varepsilon^2}{100k} +(1 +  \tfrac2{100k}) 2\varepsilon^2  \le 4\varepsilon^2.
	\]
	Applying Markov's inequality and adding the probability that the junta variables are completely separated in the partition $P_1,\ldots,P_{100k^4}$ completes the proof of the lemma.

\subsection{Proof of Lemma~\ref{lem:main-dist-estimate}}
\label{sec:second-lemma}

	For any $\bfx = (x^{(1)},\ldots,x^{(q)})$, let $\dist_{\bfx}(f_1,f_2) = \Big( \frac1q \sum_{i=1}^q \big( f_1(x^{(i)}) - f_2(x^{(i)}) \big)^2 \Big)^{1/2}$ denote the empirical distance between $f_1$ and $f_2$ according to $\bfx$. 
	To prove the lemma, we want to show that $\dist_{\bfx}(f,h)$ is within the specified bounds.
	
	The function $\dist_{\bfx}$ is a metric, so we can apply the triangle inequality to obtain
	\[
	\dist_{\bfx}(f,h) \le \dist_{\bfx}(f,g) + \dist_{\bfx}(g,h).
	\] 
	By Hoeffding's inequality, when $x^{(1)},\ldots,x^{(q)}$ are drawn independently and uniformly at random, the upper bound
	\[
	\dist_{\bfx}(f,g) \le \dist_2(f,g) + \varepsilon \le 2 \varepsilon
	\]
	holds except with probability at most $e^{-16q\varepsilon^4}$.
	
	We now want to show that $\dist_{\bfx}(g,h)$ is also close to $\dist_2(g,h)$. This analysis is a bit more subtle, however, because the choice of samples $x^{(1)},\ldots,x^{(q)}$ is \emph{not} independent of $h$ (as it affects what mapping $\psi$ will be chosen by the algorithm). So before we can apply concentration inequalities, we must ``decouple'' ${\bfx}$ and $h$. To do so, we introduce a new random process for generating ${\bfx}$. Let $\lambda : [n] \to \{0,1\}^q$ be chosen uniformly at random. This function corresponds to a random partition of the set $[n]$ of coordinates into $2^q$ parts. Let $\pi : \{0,1\}^q \to \{0,1\}^q$ be a random permutation. Then the random variable $\bfx$ obtained by setting $x^{(i)}_j = \pi( \lambda(j) )_i$ has the desired uniform distribution over sequences of $q$ vectors in $\{0,1\}^n$.
	
	This random process is designed so that the choice of $\psi$ in the algorithm (and therefore also $h$) is \emph{independent} of $\pi$; the only information about $\bfx$ used in determining it is the identity of the parts defined by $\lambda$, not what values the coordinates in each parts receive on the $q$ queries. Then
	\[
	\E_{\bfx}[ \dist_{\bfx}(g,h) ] = \E_{\lambda, r}[ \E_{\pi}[ \dist_{\bfx}(g,h)]]
	\]
	where $r$ represents the internal randomness of the algorithm outside of that used to generate $\bfx$. With probability at least $k^2/2^q$, the partition $\lambda$ completely separates the indices in $J$. Fix such a partition $\lambda$. Define $J^* = J \cup \supp(\psi)$. Then $|J^*| \le 2k$. Define $\bfy = (y^{(1)}, \ldots, y^{(q)})$ by setting $y^{(i)} = x^{(i)}_{J^*}$. Since $\dist_{\bfx}(g,h)$ only depend on the coordinates in $J^*$, we can write it equivalently as $\dist_{\bfy}(g,h)$. 

	Let $D$ denote the distribution on $\bfy$ induced by $\pi$. The distribution $D$ is close to but not equal to the uniform distribution $U$ on $\{0,1\}^{q \times |J^*|}$, since $D$ is equivalent to the distribution obtained by making drawing $(y_i^{(1)},\ldots,y_i^{(q)})$ for each $i \in J^*$ without replacement from $\{0,1\}^q$. Then
	\begin{align*}
	\Pr_{\bfy \sim D}[ |\dist_{\bfy}(g,h) - \E_{\bfy \sim U} \dist_{\bfy}(g,h)| \ge \varepsilon]
	&\le \dTV(D,U) + \Pr_{\bfy \sim U}[ |\dist_{\bfy}(g,h) - \E_{\bfy \sim U} \dist_{\bfy}(g,h)| \ge \varepsilon] \\
	&\le \tfrac{4k^2}{2^q} + e^{-16q\varepsilon^4}.
	\end{align*}
	In the last inequality, the bound $\dTV(D,U) \le \frac{(2k)^2}{2^q}$ is by the standard total variation bound between sampling with and without replacement~\cite{Freedman77} and the other bound on the other term is by Hoeffding's inequality.

\section{Applications}

In this short section, we show how Theorems~\ref{thm:submodular} and~\ref{thm:others} both follow directly from Theorem~\ref{thm:general} and the junta theorem of Feldman and \Vondrak.

\begin{proof}[Proof of Theorem~\ref{thm:submodular}]
By the first part of the Feldman--\Vondrak junta theorem, 
every submodular function $f \in \calF_n$ is 
$\frac{\epsilon}{10^6}$-close to a $k$-junta for some $k = O(\frac1{\epsilon^2}\log\frac{1}{\epsilon})$. Therefore, by Theorem~\ref{thm:general}, submodularity can be tested with $2^{O(k \log k)}/\epsilon^{10} = 2^{\tilde{O}(1/\epsilon^2)}$ queries in the $\ell_2$ testing model. By Fact~\ref{fact:lp-testing}, the number of queries for testing submodularity in the $\ell_p$ testing model for any $1 \le p < 2$ is also $2^{\tilde{O}(1/\epsilon^2)}$ and for any $p > 2$ it is $2^{\tilde{O}(1/(\epsilon^{p/2})^2)} = 2^{\tilde{O}(1/\epsilon^p)}$.
\end{proof}

\begin{proof}[Proof of Theorem~\ref{thm:others}]
By Lemma~\ref{lem:hierarchy}, additive functions, coverage functions, unit demand functions, OXS functions, and gross substitute functions are all also submodular. Therefore, the first part of the Feldman--\Vondrak junta theorem also applies to these functions and the rest of the proof is identical to that of Theorem~\ref{thm:submodular}.

Lemma~\ref{lem:hierarchy} also implies that fractionally subadditive functions are self-bounding, so the second part of the Feldman--\Vondrak junta theorem shows that every function $f$ that has either of these properties is $\frac{\epsilon}{10^6}$-close to a $k$-junta for some $k = 2^{O(\frac1{\epsilon^2})}$.
Therefore, by Theorem~\ref{thm:general}, fractional subadditivity and self-boundedness can both be tested with $2^{O(k \log k)}/\epsilon^{10} = 2^{2^{\tilde{O}(1/\epsilon^2)}}$ queries in the $\ell_2$ testing model; the general result for the $\ell_p$ testing model again follows directly from Fact~\ref{fact:lp-testing}.
\end{proof}

\bibliography{Bibliography}{}
\bibliographystyle{plain}

\appendix

\section{Missing proofs from Section~\ref{sec:prelims}}
\label{app:proofs}

We begin with the proof of Lemma~\ref{lem:infparts}. We emphasize that the proof below is essentially as found in~\cite{BWY15}; the reason we include it here is that the original statement of the proof only applied to Boolean-valued functions. As we see below, however, the same argument also holds for real-valued functions.

\begin{thm}
	\label{thm: intersecting families}
	(Dinur and Safra\cite{DS04}; Friedgut\cite{Friedgut08})
	Let $\mathcal{G}$ be a $t$-intersecting family of subsets of $[n]$ for some $t \geq 1$. For any $p < \frac{1}{t+1}$, the $p$-biased measure of $\mathcal{G}$ is bounded by $\mu_p(\mathcal{G}) \leq p^t$.
\end{thm}

\begin{proof}[Proof of Lemma~\ref{lem:infparts}]
	For $0 \leq t \leq \frac{1}{2}$, let $\mathcal{G}_t = \{J \subseteq [n]: \Inf_f(\overline{J}) < t\varepsilon^2\}$ be the family of all the sets whose compliments have influence less than $t\varepsilon^2$. For any two sets $J, K \in \mathcal{G}_\frac{1}{2}$, the subadditivity of influence implies that 
	$$
	\Inf_f(\overline{J\cap K}) = \Inf_f(\overline{J} \cup \overline{K}) \leq \Inf_f(\overline{J}) + \Inf_f(\overline{K}) < \varepsilon^2.
	$$
	
	But $f$ is $\varepsilon$-far from every  $k$-junta, so for any two sets $J, K \in \mathcal{G}_\frac{1}{2}$, $|J\cap K| > k$, from Proposition~\ref{prop:Relation_between_distance_and_influence}. Which means $\mathcal{G}_\frac{1}{2}$ is a $k + 1$ intersecting family. There are two cases now, first one is, there is at least one set $J \in \mathcal{G}_\frac{1}{2}$ such that $|J| < 2k$, second one is all the sets $J \in \mathcal{G}_\frac{1}{2}$ will have $|J| \geq 2k$. We will show that in both the cases our lemma holds. In the first case let $J \in \mathcal{G}_\frac{1}{2}$ be a set which has fewer than $2k$ elements, with high probability the set $J$ is completely separated by the partition $\mathcal{P}$, and we know that for any $K \in \mathcal{G}_\frac{1}{2}$, $|J \cap K| \geq k+1$, which means $K$ is not covered by any union of $k$-parts in $\mathcal{P}$. Therefore, $\Inf_f(\overline{J}) \geq \frac{\varepsilon^2}{2} > \frac{\varepsilon^2}{4}$ as we wanted to show.
	
	Consider the case where, all the sets in $\mathcal{G}_\frac{1}{2}$ have more than $2k$ elements. Then $G_\frac{1}{4}$ is a $2k$ intersecting family. Otherwise, if there are two sets $J, K \in \mathcal{G}_\frac{1}{4}$ such that $|J \cap K| < 2k$, then $\Inf_f(\overline{J \cap K}) \leq \Inf_f(\overline{J}) + \Inf_f(\overline{K}) < \frac{\varepsilon^2}{4} + \frac{\varepsilon^2}{4} <  \frac{\varepsilon^2}{2}$, thus contradicting our assumption.
	
	Let $J \subseteq [n]$ be the union of $k$ parts in $\mathcal{P}$. Since $\mathcal{P}$ is a random partition, $J$ is a random subset obtained by including each element of $[n]$ in $J$ independently with probability $p = \frac{k}{r} < \frac{1}{2k+1}$. By Theorem~\ref{thm: intersecting families}, $\Pr_{\mathcal{P}}[I_{f}(\overline{J}) < \frac{\varepsilon^2}{4}] = \Pr[J \in \mathcal{G}_{\frac{1}{4}}] = \mu_{\frac{k}{r}}(\mathcal{G}_{\frac{1}{4}}) \leq (\frac{k}{r})^{2k}$. By the union bound the probability that there exists a set $J \subseteq [n]$ that is the union of $k$ parts in $\mathcal{P}$ for which $\Inf_f(\overline{J}) < \frac{\varepsilon^2}{4}$ is bounded above by $\binom{r}{k}(\frac{k}{r})^{2k} \leq (\frac{er}{k})^k (\frac{k}{r})^{2k} \leq (\frac{ek}{r})^k < \frac{1}{6}$.
\end{proof}

The proof of Proposition~\ref{prop:infest} is obtained by considering the \EstInf algorithm below. \\

\begin{algorithm}[H]
	Draw $x_1,\ldots,x_m$ uniformly and independently at random from $\{0,1\}^{\overline{S}}$\;
	Draw $y_1,\ldots,y_m,y'_1,\ldots,y'_m$ uniformly and independently at random from $\{0,1\}^{S}$\;
	Return $\frac{1}{2m} \sum_{i=1}^{m} \big(f(x_i, y_i) - f(x_i,y'_i)\big)^2$\;
	\caption{\textsc{$\EstInf(f,S,m)$}}
\end{algorithm}

\bigskip

The concentration of the estimated influence is obtained via the following (standard) version of Hoeffding's inequality. 

\newtheorem*{hoeffding}{Hoeffding's inequality}
\begin{hoeffding}
	Let $X_1,..., X_n$ be independent random variables bounded by $a_1 \leq X_i \leq b_i$. Let $X = X_1 + X_2 + \cdots X_n$ have expected value $E[X] = \mu$. Then for any $t > 0$,
	$$
	\Pr[|X - \mu| \geq t] \leq 2e^{-\frac{2t^2}{\sum_{i=1}^{n}(b_i - a_i)^2}}.
	$$
\end{hoeffding}

\end{document}